\documentclass{amsart}
\usepackage{amsfonts}
\usepackage{epsfig}
\usepackage{amssymb,amsmath,bm}
\usepackage{subfig}
\usepackage{amsthm}
\usepackage{mathrsfs}
\usepackage{enumerate,color}
\theoremstyle{plain}

\newtheorem{slln}{Theorem}[section]
\newcounter{hypA}
\newenvironment{hypA}{\refstepcounter{hypA}\begin{itemize}
  \item[({\bf A\arabic{hypA}})]}{\end{itemize}}

\begin{document}
\title[Convergence of the Equi-Energy Sampler]{A note on convergence of the equi-energy sampler}

\author[Andrieu, Jasra, Doucet \& Del Moral]{By CHRISTOPHE ANDRIEU, AJAY JASRA, \\ ARNAUD DOUCET \& PIERRE DEL MORAL\\
\emph{University of Bristol, Imperial College London,}\\
\emph{University of British Columbia \& University of Nice}
}
\date{\today}
\maketitle
\begin{abstract}
In a recent paper `The equi-energy sampler with applications
statistical inference and statistical mechanics'
[\emph{Ann.~Stat.} {\bf 34} (2006)  1581--1619], Kou, Zhou \& Wong
have presented a new stochastic simulation method called the
equi-energy (EE) sampler. This technique is designed to simulate from a
probability measure $\pi$, perhaps only known up to a normalizing
constant. The authors demonstrate that the sampler performs well
in quite challenging problems but their convergence results
(Theorem 2) appear incomplete. This was pointed out, in the
discussion of the paper, by Atchad\'e \& Liu (2006) who proposed
an alternative convergence proof. However, this alternative proof,
whilst theoretically correct, does not correspond to the algorithm
that is implemented. In this note we provide a new proof of
convergence of the equi-energy sampler based on the Poisson
equation and on the theory developed in Andrieu et al.~(2007) for
\emph{Non-Linear} Markov chain Monte Carlo (MCMC). 
The objective of this note is to provide a proof of correctness of the 
EE sampler when there is only one feeding chain; the general case
requires a much more technical approach than is suitable for a short note.
In addition, we also seek to highlight the difficulties 
associated with the analysis of this type of algorithm and present the main techniques that may be adopted to prove the convergence of it.
\end{abstract}

\thispagestyle{plain}

\section{Introduction}

In this note we consider  the convergence properties of a new
stochastic simulation technique, the equi-energy sampler
introduced in (Kou, et al.~2006). This is a method designed to draw samples
from a probability measure $\pi\in\mathscr{P}(E)$ (where
$\mathscr{P}(E)$ denotes the class of probability measures) on
measurable space $(E,\mathscr{E})$, where $E$ may be a high
dimensional space and the density, is known pointwise up to a
potentially unknown constant. In particular, the algorithm
generates a non-Markovian stochastic process $\{X_n\}_{n\geq 0}$
whose stationary distribution is ultimately $\pi$; this algorithm is described
fully in Section 2.

In the paper of Kou et al.~(2006), 
an attempt to analyze the algorithm is made (in Theorem
2).
However, it was noticed in the discussion by Atchad\'e \& Liu
(2006) that this result is incomplete. We note the points that
were stated by Atchad\'e \& Liu and further expand upon their
point; see Section 3. An important remark is that Atchad\'e \& Liu attempt to
provide an alternative convergence result, via a Strong
Law of Large Numbers (SLLN) for bounded measurable functions.
Although this proof is correct, the authors study a
stochastic process which does not correspond to the algorithm; this
problem is outlined in Section 3.

The objective of this note is to provide some convergence proofs
for the EE sampler in a simple scenario (one feeding chain). We also note the difficulties 
associated with the analysis of this type of algorithm and present the main methods that can be used to prove the SLLN.
To avoid unnecessary technicalities and
focus on the `essence' of the proof, strong
assumptions are made: including the uniform ergodicity of some transition kernels.
Our proof strategy is via the Poisson equation (e.g.~Glynn \&
Meyn (1996)) and the techniques developed for Non-Linear MCMC
(Andrieu et al.~2007). That is, 
the EE sampler is a non-linear MCMC algorithm and may be analyzed in a similar
manner. Our results can be found in Section 4.

\section{Notation and Algorithm}

We now outline the notation that is adopted throughout the paper as well as the algorithm that is analyzed.

\subsection{Notation}

Define a measurable space $(E,\mathscr{E})$, with $\pi\in\mathscr{P}(E)$
(recall $\mathscr{P}(E)$ denotes the class of probability measures on $(E,\mathscr{E})$)
a target probability measure of interest.

For a stochastic process $\{X_n\}_{n\geq 0}$ on $(E^{\mathbb{N}},\mathscr{E}^{\otimes\mathbb{N}})$,
$\mathscr{G}_n=\sigma(X_0,\dots,X_n)$
is the natural filtration. $\mathbb{P}_{\mu}$ is taken as a probability law of a stochastic
process with initial distribution $\mu$ and $\mathbb{E}_{\mu}$ the associated expectation.
If $\mu=\delta_x$ (with $\delta$ the Dirac measure)
$\mathbb{P}_x$ (resp.~$\mathbb{E}_{x}$) is adopted instead of $\mathbb{P}_{\delta_x}$ (resp.~$\mathbb{E}_{\delta_x}$).
We use $X_n\stackrel{a.s}{\longrightarrow}_{\mathbb{P}}X$ to denote almost sure
convergence of $X_n$ to $X$. The equi-energy sampler generates
a stochastic process on $(\Omega,\mathscr{F})$, which is defined in the next Section.

Let $\|\eta-\mu\|_{\textrm{tv}}:=\sup_{A\in\mathscr{E}}|\eta(A)-\mu(A)|$
denote the total variation distance between $\eta,
\mu\in\mathscr{P}(E)$. Throughout, 
$K:E\rightarrow\mathscr{P}(E)$ is taken as a generic Markov kernel; the standard notations, for measurable $f:E\rightarrow\mathbb{R}$,
$K(f)(x):=\int_{E}f(y)K(x,dy)$ and for $\mu\in\mathscr{P}(E)$ 
$\mu K(f):=\int_{E}K(f)(x)\mu(dx)$ are used. Let $f:E\times
E\rightarrow\mathbb{R}$, then for $\mu\in\mathscr{P}(E)$,
$\mu(f)(x):=\int_{E}f(x,y)\mu(dy)$, with an obvious extension to
higher dimensional spaces. $\mathcal{B}_b(E)$ is used to
represent the bounded measurable functions and for
$f\in\mathcal{B}_b(E)$, $\|f\|_{\infty}:=\sup_{x\in E}|f(x)|$
is used to denote the supremum norm.

We will denote by $K_{\mu}:\mathscr{P}(E)\times E\rightarrow\mathscr{P}(E)$ a generic \emph{non-linear} Markov kernel
and its
invariant measure (given its existence) as $\omega(\mu)$ ($\omega:\mathscr{P}(E)\rightarrow\mathscr{P}(E)$).
For a sequence of probability measures $\{\mu_n\}_{n\geq 0}$ we denote the
composition $\int_{E^{n-1}} K_{\mu_1}(x,dy_1)\dots K_{\mu_n}(y_{n-1},A)$ as
$K_{\mu_1:\mu_n}(x,A)$.
The empirical
measure of an arbitrary stochastic process $\{X_n\}_{n\geq 0}$ is
defined, at time $n$, as:
\begin{eqnarray*}
S_n(du) & := & \frac{1}{n+1}\sum_{i=0}^n\delta_{x_i}(du).
\end{eqnarray*}

In addition, $a\vee b:=\max\{a,b\}$ (resp.~$a\wedge b:=\min\{a,b\}$). The indicator
function of $A\in\mathcal{E}$ is written $\mathbb{I}_A(x)$.
Note also that $\mathbb{N}_0=\mathbb{N}\cup\{0\}$, $\mathbb{T}_m:=\{1,\dots,m\}$.

\subsection{Algorithm}

We introduce a sequence of probability measures, for $r \geq 2$, $\{\pi_n\}_{n\in\mathbb{T}_{r}}$, $\pi_n\in\mathscr{P}(E)$, $n\in\mathbb{T}_r$ and $\pi_r\equiv \pi$ which are assumed to be absolutley
continuous, wrt some reference measure $\lambda^*$, and, in an abuse of notation,
write the Radon-Nikodym derivatives as $d\pi_n/d\lambda^*(x)=\pi_n(x)$ also. The EE sampler will generate a stochastic process
$\{Y_{n}^r\}_{n\geq 0}$, with $Y_n^r=(X_n^1,\dots,X_n^r)$, with $X_{n}^i:E\rightarrow \mathbb{R}^k$, $i\in\mathbb{T}_r$, $k\geq 1$ (that is $\{Y_{n}^r\}_{n\geq 0}$ is a stochastic
process on $(\Omega,\mathscr{F})=((E^r)^{\mathbb{N}},(\mathscr{E}^{\otimes r})^{\otimes\mathbb{N}})$).
Central to the construction of the EE sampler is the concept of the energy
rings; this will correspond to the partition $E=\bigcup_{i=1}^d E_i$.

For each $X_n^i$ we associate a non-linear Markov kernel $\{K_{\mu,n}\}_{n\in\mathbb{T}_{r}}$
with $K_{\mu,1}\equiv K_1$ (i.e.~$K_1$ is an ordinary Markov kernel) and
$\mu\in\mathscr{P}(E)$. Additionally, assume that
for $i=2,\dots,r-1$:
\begin{eqnarray}
\omega_{i}(\pi_{i-1})K_{\pi_{i-1},i}(dy) & = & \omega_{i}(\pi_{i-1})(dy) = \pi_{i}(dy)\label{invmeas}
\end{eqnarray}
and that $\pi_1 K_1 = \pi_1$. Here, it is assumed that, given that we input the invariant probability
measure for $K_{\pi_{i-2},i-1}$ into the non-linear kernel $K_{\mu,i}$, the target probability
measure $\pi_i$ is obtained. Define:
\begin{eqnarray}
K_{\mu,i}(x,dy) & := & (1-\epsilon) K_i(x,dy) + \epsilon Q_{\mu_x,i}(x,dy)\label{nonlinker}
\end{eqnarray}
$i=2,\dots, r$, $\epsilon \in [0,1]$, with $K_i$ a Markov kernel
of invariant distribution $\pi_i$ and also:
\begin{eqnarray*}
Q_{\mu_x,i}(x,dy) & := & \int_{E}\mu_x(dz) K^{S}_i(K_i(dy))(x,z)\\
\mu_{x}(A) & := & \sum_{i=1}^d\mathbb{I}_{E_i}(x)\frac{\mu(E_i\cap A)}{\mu(E_i)}
\end{eqnarray*}
where it is assumed $\mu(E_i)>0$; let $\mathscr{P}_{d}(E)=\{\mu\in\mathscr{P}(E):\mu(E_i) > 0~\forall i\in\mathbb{T}_d\}$.
Finally define:
\begin{eqnarray*}
K^S_i((x,y),d(x',y')) & := & \delta_{x}(dy')\delta_{y}(dx')\alpha_i(x,y) + \delta_{x}(dx')\delta_{y}(dy')[1-\alpha_i(x,y)]\\
\alpha_i(x,y) & = & 1\wedge\frac{\pi_i(y)\pi_{i-1}(x)}{\pi_i(x)\pi_{i-1}(y)}
\end{eqnarray*}
which is the swapping kernel.
It is easily seen that the kernels (\ref{nonlinker}) satisfy the equation (\ref{invmeas}). However,
it is often the case that such a system cannot be simulated exactly. The idea is to approximate the
correct probability measures $\pi_{n}$ via the empirical measures generated by the previous chain.

The algorithm which corresponds to the equi-energy sampler is as follows. Define predetermined
integers $N_1,\dots, N_r$ and assume that for all $i\in\mathbb{T}_r$, $j=\mathbb{T}_d$ (recall $d$ corresponds to the number of energy levels) we have $S^i_{N_{1:i}}(E_j)>0$ with $S^i$ the
empirical measure of the $i^{th}$ process and $N_{1:i}=\sum_{j=1}^i N_j$. The algorithm is in Figure \ref{eesampler}.

\begin{figure}[h]
\begin{flushleft}
\noindent\textsf{0.: Set $n=0$ and $X_0^{1:r}=x_0^{1:r}$, $S^l_0=\delta_{x_0^l}$, $l=1,\dots,r$. Set $i=1$.}\\
\textsf{1.: Perform the following for $i=1$ until $i=r$. Set $j=1$.}\\
\textsf{2.: Perform the following for $j=1$ until $j=N_i$, then set $i=i+1$ and go to 1.}\\
\textsf{3.: Set $n=n+1$, $k=1$.}\\
\textsf{4.: Perform the following for $k=1$ until $k=i$, then set $k=i+1$ and go to 5.}\\
$X_{n}^k\sim K_{S_{n}^{k-1},k}(x_{n-1}^k,\cdot)$, $S_n^k=S_{n-1}^k + \frac{1}{n+1}[\delta_{x_n^k} - S_{n-1}^k]$,
\textsf{set $k=k+1$ and go to 4.}\\
\textsf{5.: Perform the following for $k=i+1$ until $k\geq r$, then set $j=j+1$ and go to 2.}\\
\textsf{6.: $X_n^k\sim \delta_{x_{n-1}^k}(\cdot)$ then set $k=k+1$
and go to 5.} 
\end{flushleft}
\caption{An equi-energy sampler.} 
\label{eesampler}
\end{figure}

\noindent \emph{\textbf{Remark 1}. 
We point out here that our algorithm is slightly 
different from that of Kou et al. There, the EE
jump can be seen as using a Metropolis-Hastings (M-H) independence
sampler with proposal $\pi_{i-1}$ constrained to the set $E_i$
currently occupied by the current state (the kernel is then
approximated). We have preferred to do this in a
selection/mutation type format (see Del Moral (2004)) where a
value is selected from the empirical measure of the lower chain and
then put through a M-H exchange step. We then
allow a possibility of mutation (sampling from $K_i$).
This has been done 
in order to fit our proof in the framework of Andrieu et al.~(2007), which allows us, below, to refer to minor technical results from that work and
 hence reduce the length of this note. It should be noted that, from
a technical point of view, changing the algorithm back to the EE
sampler presents no difficulties, in terms of the following arguments. Indeed,
the only real changes to the proofs are some of the technical assumptions
in Andrieu et al.~(2007) and the uniform in time drift condition presented
there (Proposition 4.1).}

\noindent \emph{\textbf{Remark 2}. In our view, the non-linear kernel
interpretation of the equi-energy sampler allows us to
intuitively understand some practical issues associated to the
algorithm, whilst perhaps not requiring a full technical understanding. For example, if there is only one feeding chain, and 
it is stopped at some point, then we can observe from equation
(\ref{invmeas}) that this algorithm is then biased (contrary
to the point of Kou et al.~(2006) pp-1647, 5th par, although we realize that it is not possible to store an infinite number of samples).}

\section{Discussion of the Previous Proofs}

The difficulties of the convergence proofs of Kou
et al~(2006) and Atchad\'e \& Liu (2006) are now discussed.

\subsection{Theorem 2 of Kou et al.~(2006)}

We begin with the proof of Theorem 2 of Kou et al.
Recall that the Theorem states, under some assumptions,
that the steady state distribution of $\{X^{i}_n\}_{n\geq 0}$ is $\pi_i$. The authors use induction
and start by using the ergodicity of the M-H chain which verifies the case $r=1$ and continue from there.

Atchad\'e \& Liu state that equation (5) of the proof is not
clear, however, we note that the equation can indeed be verified
(and as stated by Kou et al.~(2006) in the rejoinder to the
discussion (pp-1649)) by using the SLLN (via the induction
hypothesis) and bounded convergence theorem.

The main difficulty of the proof is as follows, quoting Kou et al (2006), pp-1590:
\begin{quote}
\emph{Therefore, under the induction assumption, $X^{(i)}$ is asymptotically equivalent to a Markovian sequence
governed by $S^{(i)}(x,\cdot)$.}
\end{quote}
Here the kernel $S^{(i)}(x,\cdot)$ is the theoretical kernel 
corresponding to $K_{\pi_{i-1},i}$. The authors then state that
$S^{(i)}(x,\cdot)$ is an ergodic Markov kernel which then yields
the convergence of $X^{(i)}$. This is the difficulty of the proof:
the authors verify that the transitions of the stochastic process
are asymptotically equivalent to that of an ergodic Markov kernel, however,
this is not enough to provide the required convergence of the process.
That is, Kou et al.~(2006) prove that (suppressing the notation $N_{1:i-1}$)
\begin{equation*}
\lim_{n\rightarrow \infty}|K_{S_n^{i-1},i}(x,A)-K_{\pi_{i-1},i}(x,A)| \stackrel{a.s}
{\longrightarrow}_{\mathbb{P}^{(i-1)}}0
\end{equation*}
where $\mathbb{P}^{(i-1)}$ is the probability law of the process with $i-1$
chains. However, this convergence property essentially means that when the input probability measure $S_n^{i-1}$ is converging to the
`correct' probability measure $\pi_{i-1}$ then a set-wise convergence
of the non-linear kernel $K_{\cdot,i}$ is induced. 
This is far from sufficient as the
law of the process at iteration $n$  is, for $A\in\mathscr{E}$
$$
K_{S_1^{i-1},i}[K_{S_2^{i-1},i}[\cdots K_{S_n^{i-1},i}(A)],
$$
where $S_1^{i-1},S_2^{i-1},\ldots,S_n^{i-1}$ are empirical distributions
constructed from the same realisation of the process at level $i-1$. It is
clear that if the algorithm is to converge, then the joint distributions of $X_{n-\tau}^{(i)},\ldots,X_n^{(i)}$ for any (in fact increasing with $n$) lag $\tau$ should converge to
\[
K_{\pi_{i-1},i}\times K_{\pi_{i-1},i}\times\cdots \times K_{\pi_{i-1},i},
\]
which as we shall see is far from trivial.
This remark indicates an appropriate approach to a proof; via standard Markov chain convergence theorems.
As a result, using the arguments of Kou et al.~(2006), we cannot 
even say that
\begin{equation*}
\lim_{n\rightarrow \infty}|K_{S_{1}:S_{n+N_{1:i-1}},i}(x,A)-\pi_i(A)| \stackrel{a.s}
{\longrightarrow}_{\mathbb{P}^{(i-1)}}0
\end{equation*}
via the ergodicity of $K_{\pi_{i-1},i}(x,A)$; i.e.~a set-wise convergence of the kernel that is \emph{simulated}.


\subsection{Theorem 3.1 of Atchad\'e \& Liu (2006)}

Atchad\'e \& Liu state (pp-1625, in the proof of Theorem 3.1):
\begin{quote}
\emph{Note that the $i^{th}$ chain is actually a non-homogeneous
Markov chain with transition kernels
$K_0^{(i)},K_{1}^{(i)},\dots$, where
$K_n^{(i)}(x,A)=\mathbb{P}(X_{n+1}^{(i)}\in A|X_n^{(i)}=x)$.}
\end{quote}
This statement is not quite accurate. The $i^{th}$ chain is a
non-homogeneous Markov chain only conditional upon a realization
of the previous chain; unconditionally, it is not a Markov
chain. As a result, Atchad\'e \& Liu analyze the process of
kernel:
\begin{eqnarray*}
K_{n}^{(i)}(x,dy) & = & (1-\epsilon) K_i(x,dy) + \epsilon \mathbb{E}\bigg[R_n^{(i)}(x,dy)\bigg]
\end{eqnarray*}
where $R_n^{(i)}$ is defined in Atchad\'e \& Liu. This is not the
kernel corresponding to the algorithm; the algorithm simulates:
\begin{eqnarray*}
Q_{S^{i-1}_x,i}(x,dy) & = & \int_{E} S_x^{i-1}(dy) K^{S}_i(K_i(dy))(x,y)
\end{eqnarray*}
that is, we do not integrate over the process $\{X_n^{i-1}\}$, we
condition upon it. Therefore, the proofs of Atchad\'e \& Liu do
not provide a theoretical validation of
the equi-energy sampler.

\section{Ergodicity Results}

The SLLN is now presented: \emph{we have only proved
the case when} $r=2$ and this is assumed hereafter. There are some
difficulties in extending our proof to the case $r\geq 3$; this
will be outlined after the proofs. Note that our proof is
non-trivial and relies on a SLLN for $U-$statistics of stationary
ergodic stochastic processes (Aaronson et al.~1996).

\subsection{Assumptions}

We make the following assumptions (it is assumed that for any
$i\in\mathbb{T}_r$, $j\in\mathbb{T}_d$, $\pi_i(E_j)>0$
throughout).

\begin{hypA}\label{assump:stability}
\noindent$\bullet$ (\emph{Stability of Algorithm}): There is a
universal constant $\theta> 0$, such that for any $n\geq 0$, $j\in\mathbb{T}_d$, $i\in\mathbb{T}_{r-1}$ we have, recalling that $N_{1:i}=\sum_{j=1}^i N_j$:
\begin{eqnarray*}
S_{N_{1:i}+n}^i(E_j) & \geq & \theta \qquad \mathbb{P}_{x_0^{1:r}}-a.s.
\end{eqnarray*}
\end{hypA}

\begin{hypA}\label{assump:KandP}
\noindent$\bullet$ (\emph{Uniform Ergodicity}): The $\{K_n\}_{n\in\mathbb{T}_r}$
are uniformly ergodic Markov kernels with a one step minorization condition.
That is: $\forall n\in\mathbb{T}_r$, $\exists (\phi_n,\nu_n)\in \mathbb{R}^+\times\mathscr{P}(E)$
such that $\forall (x,A) \in E\times \mathscr{E}$:
\begin{eqnarray*}
K_n(x,A) & \geq & \phi_n\nu_n(A).
\end{eqnarray*}
\end{hypA}

\begin{hypA}\label{assump:statespace}
\noindent$\bullet$ (\emph{State-Space Constraint}):
$E$ is polish (separable complete metrisable topological space).
\end{hypA}

\subsection{Discussion of Assumptions}

The assumptions we make are quite strong. The first assumption (A\ref{assump:stability}) is used to allow us
to bound:
\begin{displaymath}
\frac{1}{S^i_{m+1}(E_i)} - \frac{(m+2)}{(m+1)S^i_{m}(E_i)}
\end{displaymath}
which will appear in the proof below.
This assumption, on the empirical measure, is removed in Andrieu et al.~(2007);
however, this is at the cost of a significant increase in the technicalities
of the proof. As a result, (A\ref{assump:stability}) is adopted as an intuitive
assumption as it states:
\begin{enumerate}
\item{Make sure that $\pi_i(E_j)$ for all $i, j$ is non-negligable.}
\item{Let $N_1,\dots, N_{r-1}$ be reasonably large so that we can expect convergence.}
\end{enumerate}

The second assumption (A\ref{assump:KandP}) might appear strong, but allows us to 
significantly simplify both notation and our proofs whilst preserving the `essence' of the general proof. In addition, this condition will often be satisfied on finite state spaces. More general assumptions could be 
used, at the expense of significant notational and technical complexity.
The assumption allows us to use the following facts:
\begin{enumerate}
\item{For any fixed $\mu\in\mathscr{P}_d(E)$, $\exists \omega_i(\mu)\in\mathscr{P}(E)$ such that $\omega_i(\mu)K_{\mu,i}=\omega_i(\mu)$.}
\item{For any fixed $\mu\in\mathscr{P}_d(E)$, $i\in\mathbb{T}_r$, $\exists\rho \in(0,1)$, $M<\infty$ such that for any $n\in\mathbb{N}$
we have $\sup_{x\in E}\|K_{\mu,i}^n(x,\cdot)-\omega_i(\mu)\|_{\textrm{tv}}\leq M\rho^n$.}
\end{enumerate}
These properties will help to simplify our proofs below.

The final assumption (A\ref{assump:statespace}) will be related to
some technical arguments in the proof.

\subsection{SLLN}

We are to establish the convergence of $S_{n}^r(f)\stackrel{a.s}{\longrightarrow}_{\mathbb{P}_{x_{0}^{1:r}}}\pi_r(f)$
for some $f$ to be defined in the proof and $n\geq N_{1:r-1}$.

\subsubsection{Strategy of the Proof}

Our approach is to consider
$S^{\omega}_{n,r}= 1/(n-N_{1:r-1}+1)\sum_{j=N_{1:r-1}}^n\omega_r(S_j^{r-1})$
and adopt the decomposition:
\begin{eqnarray}
S_n^r(f) -\pi_r(f) & = & S_n^r(f) - S^{\omega}_{n,r}(f) + S^{\omega}_{n,r}(f) -\pi_r(f)\label{eq:prfdecomp}.
\end{eqnarray}
The analysis of the first term on the RHS of (\ref{eq:prfdecomp})
relies upon a Martingale argument using the classical
Poisson's equation solution:
\begin{eqnarray*}
f(X_n^r)-\omega(S_n^{r-1})(f) & = & \hat{f}_{S_n^{r-1}}^{r}(X_n^r) - K_{S_n^{r-1},r}(\hat{f}_{S_n^{r-1}}^{r})(X_n^r)
\end{eqnarray*}
where $\hat{f}_{S_n^{r-1}}^{r}$ is a solution of the Poisson equation.
Indeed, the first term on the RHS of (\ref{eq:prfdecomp}) can be rewritten:
\begin{eqnarray}
(n-N_{1:r-1}+1)[S_n^i-S_{n,r}^{\omega}](f) & = & M_{n+1}^r + \label{eq:decomp}
 \sum_{m=N_{1:r-1}}^{n}[\hat{f}_{S_{m+1}^{r-1}}^r(X_{m+1}^r) -
\\ & & \hat{f}_{S_m^{r-1}}^r(X_{m+1}^r)] +  \hat{f}_{S_{N_{1:r-1}}^{r-1}}^r(X_0^r)
 - \hat{f}_{S_{n+1}^{r-1}}^r(X_{n+1}^r) \nonumber
\end{eqnarray}
where
\begin{eqnarray}
M_{n+1}^r & = & \sum_{m=N_{1:r-1}}^{n}[\hat{f}_{S_m^{r-1}}^r(X_{m+1}^r) - K_{S_m^{r-1},r}(\hat{f}_{S_m^{r-1}}^r)(X_m^r)] \nonumber\\
\hat{f}_{S_m^{r-1}}^r(X_{m+1}^r) & = &
\sum_{n\in\mathbb{N}_0}[K^{n}_{S_m^{r-1},r}(f)(X_{m+1}^r) -
\omega_r(S_m^{r-1})(f)]\label{eq2}
\end{eqnarray}
and $\{M_n^r,\mathscr{G}_n\}_{n\geq 0}$ is a martingale and $M_n^r:=0$, for
$0\leq n\leq N_{1:r-1}$.
Recall that (\ref{eq2}) is a solution to the Poisson equation,
which will exist under our assumptions above. 

The proof will deal
with the Martingale via the Burkh\"older inequality and the fluctuations
of the solution of the Poisson equation due to the evolution of the empirical
measure (\ref{eq2}) using continuity properties of the kernel $Q_{\mu}$.
The bias term $S^{\omega}_{n,r}(f) -\pi_r(f)$ is controlled by a SLLN for
$U-$statistics of stationary ergodic stochastic processes.

\subsubsection{Main Result}

\begin{slln}
Assume (A\ref{assump:stability}-\ref{assump:KandP}). Then for any $p\geq 1$, $\exists B_p<\infty$ such that for any $n\geq N_{1:r-1}$ and $f\in\mathcal{B}_b(E)$ we have that:
\begin{eqnarray*}
\mathbb{E}_{x_0^{1:r}}\bigg[|[S_{n}^r - S_{n,r}^{\omega}](f)|^p\bigg]^{1/p} & \leq & \frac{B_p\|f\|_{\infty}}{(n - N_{1:r-1} + 1)^{\frac{1}{2}}}.
\end{eqnarray*}
if, in addition, (A\ref{assump:statespace}) holds then
for any $f\in\mathcal{B}_b(E)$:
\begin{eqnarray*}
S_{n}^2(f)\stackrel{a.s}{\longrightarrow}_{\mathbb{P}_{x_{0}^{1:2}}}
\pi_2(f).
\end{eqnarray*}
\end{slln}

\begin{proof}
Our proof relies heavily upon the theory of Andrieu et al.~(2007).
Note that, under (A2) and, for any fixed $\mu\in\mathscr{P}_d(E)$, 
the uniform ergodicity of the kernel $K_{\mu,i}$
allows us to use
the methods of Andrieu et al.~(2007).
We will follow the proof of Theorem 6.5 of that paper. 
In order
to prove the SLLN in the paper, the authors combine a series of
technical results. The first of which is the Lipschitz
continuity of the kernel $Q_{\mu}$; we establish the result for
bounded functions and the particular kernel considered here. To
simplify the notation, we remove the sub/superscripts from the
various objects below.

Let $f\in\mathcal{B}_b(E)$ and $\mu,\xi\in\mathscr{P}_d(E)$, then we have:
\begin{eqnarray*}
|Q_{\mu_x}(f)(x)  - Q_{\xi_x}(f)(x)| & = & \sup_{(x,y)\in E^2}\|K^S(K(f))(x,y) - Q_{\mu_x}(f)(x)\|_{\infty} \\ & &
\times\bigg|\int_{E\times E}\frac{K^S(K(f))(x',y) - Q_{\mu_x}(f)(x')}{\sup_{(x,y)\in E^2}\|K^S(K(f))(x,y) - Q_{\mu_x}(f)(x)\|_{\infty}}
\times \\ & &
\mu_x(dy)\times\delta_x(dx') - \xi_x(dy)\times\delta_x(dx')\bigg|\\
& \leq & 2\|f\|_{\infty} \sup_{x\in E}\|\mu_x - \xi_x\|_{\textrm{tv}}
\end{eqnarray*}

We then note that Propositions 6.1 and 6.2 (bounding the solution
of the Poisson equation and Martingale in the $\mathbb{L}_p$ norm)
of Andrieu et al.~(2007) are proved in the same manner. That is, in a similar
way to the proofs constructed there, we can show that:
\begin{eqnarray*}
\mathbb{E}_{x_0^{1:r}}\bigg[|\widehat{f}_{S_m}(X_{m+1})|^p\bigg]^{1/p}
& \leq & M\|f\|_{\infty}\\
\mathbb{E}_{x_0^{1:r}}\big[|M_n|^p\big]^{1/p} & \leq & M\|f\|_{\infty}n^{1/2}.
\end{eqnarray*}
As a result, the verification of Proposition 6.3 (bounding the
fluctuations of the Poisson equation due to the evolution of the empirical measure) and
Theorem 6.5 (the SLLN) are required.

We begin with the equation (\ref{eq2}); the bound is proved by establishing:
\begin{eqnarray}
|S_{m+1,x}(f) - S_{m,x}(f)| & \leq & \frac{M\|f\|_{\infty}}{m+2}\label{eq:prf1}
\end{eqnarray}
for $M<\infty$ some constant and any $f\in\mathcal{B}_b(E)$.
Consider
\begin{eqnarray*}
|S_{m+1,x}(f) - S_{m,x}(f)| & = & \bigg|\sum_{i=1}^d\mathbb{I}_{E_i}(x)\bigg[\frac{S_{m+1}(\mathbb{I}_{E_i}f)}{S_{m+1}(E_i)} - \frac{S_{m}(\mathbb{I}_{E_i}f)}{S_{m}(E_i)}\bigg] \bigg|\\
& = & \bigg|\sum_{i=1}^d\mathbb{I}_{E_i}(x)\bigg[\frac{f(x_{m+1})\mathbb{I}_{E_i}(x_{m+1})}{(m+2)S_{m+1}(E_i)}
+ \frac{1}{m+2}\sum_{j=0}^m f(x_j)\mathbb{I}_{E_i}(x_{j})\\ & & \times \bigg\{\frac{1}{S_{m+1}(E_i)} - \frac{m+2}{(m+1)S_{m}(E_i)}\bigg\}
\bigg]\bigg|.
\end{eqnarray*}
Now, since:
\begin{eqnarray*}
\bigg|\frac{1}{S_{m+1}(E_i)} - \frac{(m+2)}{(m+1)S_{m}(E_i)}\bigg| & = & \frac{|(m+1)S_m(E_i) - \delta_{x_{m+1}}(E_i) - (m+1)S_{m}(E_i)|}{(m+1)S_m(E_i)S_{m+1}(E_i)}\\
& = & \frac{\delta_{x_{m+1}}(E_i)}{(m+1)S_m(E_i)S_{m+1}(E_i)}\\
& \leq & \frac{1}{(m+1)\theta^2}
\end{eqnarray*}
it follows that:
\begin{eqnarray*}
|S_{m+1,x}(f) - S_{m,x}(f)| & \leq & \frac{\|f\|_{\infty}}{(m+2)}
\sum_{i=1}^d\mathbb{I}_{E_i}(x)\bigg[\frac{1}{\theta}
+\frac{1}{\theta^2}\bigg] \\
& \leq & \frac{M\|f\|_{\infty}}{m+2}
\end{eqnarray*}
as required.

To bound the fluctuations of the Poisson equation, the decomposition (Proposition B.5) in Andrieu et al.~(2007) is adopted, along with Minkowski's inequality:
\begin{displaymath}
\mathbb{E}_{x_{0}^{1:r}}\bigg[|\widehat{f}_{S_{m+1}}(X_{m+1}) -
\widehat{f}_{S_{m}}(X_{m+1})|^p\bigg]^{1/p}  \leq
\end{displaymath}
\begin{displaymath}
\sum_{n\in\mathbb{N}_0}\bigg[\mathbb{E}_{x_{0}^{1:r}}\bigg[|\sum_{i=0}^{n-1}[K_{S_{m+1}}^i
- \omega(S_{m+1})](K_{S_{m+1}} - K_{S_{m}})[K_{S_m}^{n-i-1}-\omega(S_m)](f)(X_{m+1})|^p\bigg]^{1/p}
+
\end{displaymath}
\begin{displaymath}
\mathbb{E}_{x_{0}^{1:r}}\bigg[|[\omega(S_{m+1})-\omega(S_{m})](K_{S_m}^n-\omega(S_{m}))|^p\bigg]^{1/p}\bigg]
\end{displaymath}
To bound the first expression on the RHS, we can use the fact
that, for a fixed (deterministic) pair of empirical measures
$S_m,S_{m+1}\in\mathscr{P}_d(E)$ and for any $x\in E$:
\begin{displaymath}
|[K_{S_{m+1}}^i
- \omega(S_{m+1})](K_{S_{m+1}} - K_{S_{m}})[K_{S_m}^{n-i-1}-\omega(S_m)](f)(x)|
\leq \end{displaymath}
\begin{displaymath}
M\rho^{i}\|(K_{S_{m+1}} - K_{S_{m}})[K_{S_m}^{n-i-1}-\omega(S_m)](f)\|_{\infty}
\end{displaymath}
and further, for any $x\in E$:
\begin{displaymath}
|(K_{S_{m+1}} - K_{S_{m}})[K_{S_m}^{n-i-1}-\omega(S_m)](f)(x)|
\leq \frac{M\|[K_{S_m}^{n-i-1}-\omega(S_m)](f)\|_{\infty}}{m+2}
\end{displaymath}
due to the Lipschitz continuity of $Q$ and the bound (\ref{eq:prf1}); therefore:
\begin{eqnarray*}
\|[K_{S_{m+1}}^i
- \omega(S_{m+1})](K_{S_{m+1}} - K_{S_{m}})[K_{S_m}^{n-i-1}-\omega(S_m)](f)\|_{\infty}
& \leq & \frac{M\rho^{n-1}}{m+2}.
\end{eqnarray*}
Since, due to (A\ref{assump:stability}), this property holds almost surely,
it is possible to bound the first expression. The second expression is dealt
with in a similar manner, using the inequality (see Andrieu et al.~(2007)):
\begin{eqnarray*}
\|[\omega(S_{m+1})-\omega(S_{m})](f)\|_{\infty} & \leq &
M\|[K_{S_{m+1}}-K_{S_{m}}](f)\|_{\infty}.
\end{eqnarray*}
This result can be obtained by the continuity of invariant measures
of uniformly ergodic Markov kernels indexed by a parameter.

To complete the first part of the proof, we can use the manipulations of Del Moral \& Miclo (2004), Proposition 3.3, to yield:
\begin{eqnarray*}
\mathbb{E}_{x_0^{1:r}}\bigg[|[S_{n}^r - S_{n,r}^{\omega}](f)|^p\bigg]^{1/p} & \leq & \frac{B_p\|f\|_{\infty}}{(n - N_{1:r-1} + 1)^{\frac{1}{2}}}.
\end{eqnarray*}

To control the bias $S^{\omega}_{n,r}(f) -\pi_r(f)$ when
$r=2$, the following decomposition is adopted:
\begin{eqnarray*}
|[\omega(S_m) - \omega(\pi_{1})](f)| & \leq &  |[K_{S_m}^q - K_{\pi_1}^q](f)|
+ |[\omega(S_m) - K_{S_m}^q](f)| +|[K_{\pi_1}^q - \omega(\pi_{1})](f)|.
\end{eqnarray*}
Due to the uniform ergodicity bound $\|K_{\mu}^q-\omega(\mu)\|_{\textrm{tv}}\leq M\rho^k$
we will show that for any $q\in\mathbb{N}$:
\begin{eqnarray}
\lim_{m\rightarrow\infty}|[K_{S_m}^q - K_{\pi_1}^q](f)| & = & 0 \qquad \mathbb{P}_{x_{0}^{1:r}}-a.s.
\label{eq:prfeq2}
\end{eqnarray}

Let $\epsilon=1$; the general case is dealt with below.
Let $\mu\in\mathscr{P}_d(E)$, and for simplicity write $K^S(K(f)\times 1)(x,y):=
P(f)(x,y)$, $f\in\mathcal{B}_b(E)$, then we will prove by induction that:
\begin{equation}
\underbrace{Q_{\mu_x}Q_{\mu_{\cdot}}\dots Q_{\mu_{\cdot}}}_{q~\textrm{times}}(f)(x)
 = \sum_{(i_1,\dots,i_q)\in\mathbb{T}_d^q}
\frac{\mathbb{I}_{E_{i_1}}(x)}
{\prod_{j=1}^q\mu(E_{i_{j}})}\mu^{\otimes
q}\bigg\{\big(\prod_{j=1}^q\mathbb{I}_{E_{i_j}}\big)
\underbrace{P(\mathbb{I}_{E_{i_2}}P(\mathbb{I}_{E_{i_3}}\cdots
P(\mathbb{I}_{E_{i_q}}P(f))))}_{q-1~\textrm{terms}} \bigg\}(x)
\label{eq:prfeq1}
\end{equation}
where a composition of the $P$ kernels is defined as:
\begin{equation*}
P^q(f)(x,x_{1:q}) := \int_{E^{q+1}}P((x,x_1),dy_1)P((y_1,x_2),dy_2)\dots P((y_{q-1},x_q),dy_q)f(y_q).
\end{equation*}
For $q=1$ (\ref{eq:prfeq1}) clearly holds, so assume for $q-1$ and consider
$q$:
\begin{eqnarray*}
Q_{\mu_x}Q_{\mu_{\cdot}}\dots Q_{\mu_{\cdot}}(f)(x) & = &
\sum_{(i_1,\dots,i_{q-1})\in\mathbb{T}_d^{q-1}}
\frac{\mathbb{I}_{E_{i_1}}(x)}
{\prod_{j=1}^{q-1}\mu(E_{i_{j}})}\mu^{\otimes (q-1)}\bigg\{\big(\prod_{j=1}^{q-1}\mathbb{I}_{E_{i_j}}\big)\\
& &
P(\mathbb{I}_{E_{i_2}}\cdots
P(\mathbb{I}_{E_{i_{q-1}}}P(Q_{\mu_{\cdot}}(f))))\bigg\}(x)
\end{eqnarray*}
To continue the proof, consider:
\begin{eqnarray*}
P(Q_{\mu_{\cdot}}(f))(x,x_1) & = & \int_E P((x,x_1),dy_1)\int_{E}\mu_{y_1}(dx_2)P(f)(y_1,x_2)\\
& = & \int_E P((x,x_1),dy_1) \sum_{i=1}^d\mathbb{I}_{E_i}(y_1)\frac{\int\mathbb{I}_{E_i}(x_2)P(f)(y_1,x_2)\mu(dx_2)}{\mu(E_i)}\\
& = & \sum_{i=1}^d\frac{\mu(\mathbb{I}_{E_i}P(\mathbb{I}_{E_i}P(f)))}{\mu(E_i)}.
\end{eqnarray*}
Thus, due to the above equation:
\begin{eqnarray*}
Q_{\mu_x}Q_{\mu_{\cdot}}\dots Q_{\mu_{\cdot}}(f)(x) & = & \sum_{(i_1,\dots,i_{q-1})\in\mathbb{T}_d^{q-1}}
\frac{\mathbb{I}_{E_{i_1}}(x)}
{\prod_{j=1}^{q-1}\mu(E_{i_{j}})}\mu^{\otimes (q-1)}\bigg\{\big(\prod_{j=1}^{q-1}\mathbb{I}_{E_{i_j}}\big)\\
& &
P(\mathbb{I}_{E_{i_2}}\cdots
P\bigg[\mathbb{I}_{E_{i_{q-1}}}\sum_{i_q=1}^d\frac{\mu(\mathbb{I}_{E_{i_q}}P\mathbb{I}_{E_{i_q}}P(f)))}{\mu(E_{i_q})}\bigg])\bigg\}(x)\\
& = & \sum_{(i_1,\dots,i_{q})\in\mathbb{T}_d^{q}}
\frac{\mathbb{I}_{E_{i_1}}(x)}
{\prod_{j=1}^{q}\mu(E_{i_{j}})}\mu^{\otimes (q-1)}\bigg\{\big(\prod_{j=1}^{q-1}\mathbb{I}_{E_{i_j}}\big)\\
& &
P(\mathbb{I}_{E_{i_2}}\cdots
P\bigg[\mathbb{I}_{E_{i_{q-1}}}\mu(\mathbb{I}_{E_{i_q}}P\mathbb{I}_{E_{i_q}}P(f)))\bigg])\bigg\}(x).
\end{eqnarray*}
Application of Fubini's theorem yields the desired result.

To prove, for $\epsilon=1$, that (\ref{eq:prfeq2}) holds, observe that:
\begin{eqnarray*}
|[K_{S_m}^q - K_{\pi_1}^q](f)| & = &  \bigg|\sum_{(i_1,\dots,i_q)\in\mathbb{T}_d^q}\bigg[
\frac{\mathbb{I}_{E_{i_1}}(x)}
{\prod_{j=1}^q S_m(E_{i_{j}})}S_m^{\otimes q}\bigg\{\big(\prod_{j=1}^q\mathbb{I}_{E_{i_j}}\big)
P(\mathbb{I}_{E_{i_2}}P(\mathbb{I}_{E_{i_3}}\cdots
P(\mathbb{I}_{E_{i_q}}P(f))))
\bigg\}(x)- \\ & &
\frac{\mathbb{I}_{E_{i_1}}(x)}
{\prod_{j=1}^q \pi_1(E_{i_{j}})}\pi_1^{\otimes q}\bigg\{\big(\prod_{j=1}^q\mathbb{I}_{E_{i_j}}\big)
P(\mathbb{I}_{E_{i_2}}P(\mathbb{I}_{E_{i_3}}\cdots
P(\mathbb{I}_{E_{i_q}}P(f))))
\bigg\}(x)
\bigg]\bigg|.
\end{eqnarray*}

Application of Theorem U and Proposition 2.8 of Aaronson et al.~(1996) (along with the Theorem
for almost sure convergence of continuous transformations of almost surely
convergent random variables) yields the desired result. 
Firstly, note that these are
results associated to the almost sure convergence of $U-$ and $V-$ (Von Mises) statistics; this is where (A\ref{assump:statespace}) is required. 
Secondly, we remark
that it is not required that the auxiliary process is started in its stationary
regime (as stated in the result of Aaronson et al.~(1996)): We can adopt
a coupling argument for uniformly ergodic Markov chains, along the lines of Andrieu et al.~(2007) (Theorem 6.5 and Proposition C.1).
To complete
the proof for $\epsilon\in(0,1)$, we note the following decomposition for
iterates of mixtures of Markov kernels $K$ and $P$:
\begin{equation*}
((1-\epsilon)K+\epsilon P)^n(x,dy) = \sum_{l=0}^n\epsilon^l(1-\epsilon)^{n-l}
\sum_{(\alpha_1,\dots,\alpha_n)\in\mathcal{S}_l}
K^{1-\alpha_1}P^{\alpha_1}\dots K^{1-\alpha_n} P^{\alpha_n}(x,dy).
\end{equation*}
where $\mathcal{S}_l=\{(\alpha_1,\dots,\alpha_n):\sum_{j=1}^n\alpha_j=l\}$;
there is no difficulty to extend the result, using the bounded convergence
theorem where required.
\end{proof}

\noindent \emph{\textbf{Remark 3}. In the proof we have adopted a decomposition
that has naturally led to the use of SLLN for $U-$statistics. Essentially,
the algorithm requires that the invariant measures converge to the desired
distribution, and this is manifested, in our proof, via the iterates of
the non-linear kernel. This is the main difficulty in proving the SLLN
for the equi-energy sampler.
An alternative
approach, via uniform SLLN, may also be adopted, possibly at the cost of more
abstract assumptions; see Del Moral (2004) for example, in the case of particle
approximations of Feynman-Kac formulae.}

\noindent \emph{\textbf{Remark 4}. 
We note that it is possible to extend our proof, via a density argument (see Del Moral (1998)), for a related algorithm, (NL3) of Andrieu et al.~(2007), with $r-1$ feeding
chains, but that this cannot be used for the equi-energy sampler, due to
the fact that the indicator functions in the definition of the kernel (\ref{nonlinker})
are not continuous. In general, a proof by induction requires
more complicated arguments and
as
a result, we feel that the convergence of the equi-energy sampler, as well
as the convergence rate (as brought up in the discussion of Kou et al.~(2006))
are non-trivial research problems.}

\vspace{0.05 in}

{\ \nocite{*} \centerline{ REFERENCES}
\begin{list}{}{\setlength{\itemindent}{-0.3in}}
\item {\sc Aaronson}, J., {\sc Burton}, R., {\sc Dehling}, H.,
{\sc Gilhat}, D., {\sc Hill}, T. \& {\sc Weiss} B.~(1996). Strong
laws for $L-$ and $U-$ statistics. {\it Trans. Amer. Math. Soc.},
{\bf 348}, 2845--2866. 
\item
{\sc Andrieu,} C., {\sc Jasra}, A., {\sc Doucet}, A. \& {\sc Del
Moral} P.~(2007). Non-Linear Markov chain Monte Carlo via self
interacting approximations. Technical Report, University of
Bristol. \item {\sc Atchad\'e}, Y. \& {\sc Liu}, J. S.~(2006).
Discussion of Kou, Zhou \& Wong. {\it Ann. Statist.}, {\bf 34},
1620--1628. \item {\sc Del Moral}, P.~(1998). Measure valued
processes and interacting particle systems. Application to non
linear filtering problems. {\it Ann. Appl. Prob.}, {\bf 8},
438--495. \item {\sc Del Moral}, P.~(2004). \textit{Feynman-Kac
Formulae: Genealogical and Interacting Particle Systems with
Applications}. Springer: New York. \item {\sc Del Moral}, P. \&
{\sc Miclo} L.~(2004). On convergence of chains with occupational
self-interactions. {\it Proc. R. Soc. Lond. A}, {\bf 460},
325--346. \item {\sc Glynn,} P. W. \& {\sc Meyn} S. P.~(1996). A
Liapounov bound for solutions of the Poisson equation. {\it Ann.
Prob.}, {\bf 24}, 916--931. 
\item {\sc Kou}, S. C, {\sc Zhou}, Q., \& {\sc Wong}, W.
H.~(2006). Equi-energy sampler with applications to statistical
inference and statistical mechanics (with discussion). {\it Ann.
Statist.}, {\bf 34}, 1581--1619.
\end{list}

\vspace{0.2cm}
\noindent
{\sc Department of Mathematics}\hspace{4cm}  {\sc Department of Mathematics}\\
{\sc University of Bristol}\hspace{5.2cm} {\sc Imperial College London}\\
{\sc Bristol}\hspace{7.8cm} {\sc London}\\
{\sc England}\hspace{7.6cm} {\sc England}\\
{\sc E-Mail:} c.andrieu@bris.ac.uk\hspace{4.5cm} {\sc E-mail:}a.jasra@ic.ac.uk\\

\noindent
{\sc Department of Statistics}\hspace{4.6cm}{\sc Department of Mathematics}\\
{\sc University of British Columbia}
\hspace{3.3cm}{\sc University of Nice}\\
{\sc Vancouver}\hspace{7.1cm} {\sc Nice}\\
{\sc Canada}\hspace{7.7cm} {\sc France}\\
{\sc E-Mail:} arnaud@stat.ubc.ca\hspace{4.5cm} {\sc E-mail:}delmoral@math.unice.fr\\
\end{document}